\documentclass[a4paper,11pt]{llncs}
\usepackage[english]{babel} 
\usepackage[T1]{fontenc}   
\usepackage[latin1]{inputenc}
\usepackage{amsmath}
\usepackage{amsfonts}
\usepackage{amssymb}
\usepackage{mathrsfs}
\usepackage{url}
\usepackage{enumerate}
\usepackage{algorithm}
\usepackage{algorithmic}
\usepackage{color}
\usepackage{hyperref}
\spnewtheorem{heuristic}{Heuristic}{\bfseries}{\itshape}
\spnewtheorem{assumption}[theorem]{Assumption}{\bfseries}{}
\spnewtheorem{nota}{Notation}{\bfseries}{}
\spnewtheorem{fact}{Fact}{\bfseries}{}

\title{An efficient structural attack
  on NIST submission DAGS}
\date{}
\author{}
\author{\'Elise Barelli\inst{1} \and Alain Couvreur\inst{1}}
\institute{INRIA \&
LIX, CNRS UMR 7161,\\
\'Ecole polytechnique, 91128 Palaiseau Cedex, France.\\
 \email{elise.barelli@inria.fr},
 \email{alain.couvreur@lix.polytechnique.fr}
}

\newcommand{\tr}{\ensuremath{\textrm{Tr}}}
\newcommand{\nr}{\ensuremath{\textrm{N}}}
\newcommand{\Z}{\mathbb{Z}}

\newcommand{\eqdef}{\stackrel{\text{def}}{=}}

\renewcommand{\leq}{\leqslant} 
 
\renewcommand{\geq}{\geqslant}

\newcommand{\Fq}{\ensuremath{\mathbb{F}_q}}
\newcommand{\Fqq}{\ensuremath{\mathbb{F}_{q^2}}}
\newcommand{\Fqm}{\ensuremath{\mathbb{F}_{q^m}}}

\newcommand{\F}{\ensuremath{\mathbb{F}}}

\newcommand{\code}[1]{\ensuremath{\mathscr{#1}}}
\newcommand{\Cpub}{{\code{C}_{\text{pub}}}}

\newcommand{\AC}{\code{A}}
\newcommand{\BC}{\code{B}}
\newcommand{\CC}{\code{C}}
\newcommand{\DC}{\code{D}}

\newcommand{\SC}{\code{S}}

\newcommand{\EC}{\code{E}}
\newcommand{\XC}{\code{X}}
\newcommand{\ZC}{\code{Z}}
\newcommand{\NTC}{\code{N}{\!\!}\code{T}}

\newcommand{\spc}[2]{#1 \star #2}

\newcommand{\cset}[3]{\left\{\left.\left(#1\right)_{#2} ~\right|~#3 \right\}}

\newcommand{\sq}[1]{#1^{\star 2}}

\newcommand{\cond}[2]{\mathbf{Cond}(#1, #2)}

\newcommand{\av}{\mat{a}}
\newcommand{\bv}{\mat{b}}
\newcommand{\cv}{\mat{c}}

\newcommand{\uv}{\mat{u}}

\newcommand{\xv}{{\mat{x}}}

\newcommand{\yv}{{\mat{y}}}

\newcommand{\onev}{{\mat{1}}}

\newcommand{\mat}[1]{\ensuremath{\boldsymbol{#1}}}

\newcommand{\Gm}{\mat{G}}
\newcommand{\Hm}{\mat{H}}

\newcommand{\Ind}{\mathcal{I}}

\newcommand{\GRS}[3]{\text{\bf GRS}_{#1}(#2,#3)}
\newcommand{\RS}[2]{\text{\bf RS}_{#1}(#2)}
\newcommand{\Alt}[3]{\code{A}_{#1}(#2,#3)}

\newcommand{\sh}[2]{\mathcal{S}_{#2}\left(#1\right)}
\newcommand{\pu}[2]{\mathcal{P}_{#2}\left(#1\right)}

\newcommand{\loc}[1]{\pi_{#1}}

\newcommand{\map}[4]{\left\{
\begin{array}{ccc}
#1 & \longrightarrow & #2 \\
#3 & \longmapsto     & #4
        \end{array}
\right.}

\newcommand{\group}{\mathcal{G}}
\newcommand{\inv}[2]{{#1}^{#2}}
\newcommand{\punctinv}[2]{\overline{#1}^{#2}}

\newcommand{\dags}[1]{\texttt{DAGS\_{#1}}}

\begin{document}

\maketitle

% \begin{center}
%   \includegraphics[scale=.25]{dagger}
% \end{center}

\begin{abstract}
  We present an efficient key recovery attack on code based
  encryption schemes using some quasi--dyadic alternant codes with
  extension degree $2$. This attack permits to break the proposal DAGS
  recently submitted to NIST.
\end{abstract}


\section*{Introduction}
In 1978, in the seminal article \cite{M78}, R.~J.~McEliece designed a
public key encryption scheme relying on the hardness of the bounded
decoding pro\-blem~\cite{BMT78}, {\em i.e.} on the hardness of decoding
an arbitrary code. For a long time, this scheme was considered as
unpractical because of the huge size of the public keys compared to
public key encryption schemes relying on algorithmic number theoretic
problems.
% Hence, for a long time, code based cryptography was considered as a
% purely theoretic area with few pers\-pectives of practical applications.
The trend changed in the last decade because of the progress of
quantum computing and the increasing threat of the existence in a near
future of a quantum computer able to break usual cryptography
primitives based on number theoretic pro\-blems.
% : integer factorisation
% and discrete logarithm in finite fields or elliptic curves.
An evidence for this change of trend is the recent call of the
National Institute for Standards and Technology (NIST) for post
quantum cryptography. The majority of the submissions to this call are
based either on codes or on lattices.

After forty years of research on code based cryptography, one can
identify two general trends for instantiating McEliece's scheme.  The
first one consists in using codes from probabilistic constructions
such as MDPC codes \cite{MTSB13,BBC13}. The other one
consists in using algebraic codes such as Goppa codes or more
generally alternant codes. A major difference between these two
families of proposals is that the first one, based on MDPC codes
benefits in some cases from clean security reductions to the
decoding problem.

Concerning McEliece instantiations based on algebraic codes,
which include McEliece's original proposal based on binary Goppa
codes, two approaches have been considered in order to address the
drawback of the large of pubic key sizes.
On the one hand, some proposals suggested to replace Goppa or
alternant codes by more structured codes such as generalised
Reed--Solomon (GRS) codes \cite{N86}, their low dimensional subcodes
\cite{BL05}, or GRS codes to which various transformations have been
applied \cite{W06,BBCRS14,W16}. It turns out that most of these
proposals have been subject to polynomial time key-recovery attacks
\cite{SS92,W10,CGGOT14,COTG15}.  In addition, proposals based on Goppa codes
which are {\em close} to GRS codes, namely Goppa code with a low
extension degree $m$ have been the target of some structural attacks
\cite{FPP14,COT17}.
% such as distinguishing filtration attacks against wild Goppa codes
% with extension degree $m=2$ \cite{COT17} and algebraic attacks based
% on multivariate polynomial system solving against wild Goppa codes
% with extension degree $m\leq 3$ over non prime fields \cite{FPP14}.
On the other hand, many proposals suggest the use of codes with a non
trivial automorphism group \cite{G05,BCGO09,MB09,P12}.  A part of
these proposals has been either partially or completely broken
\cite{OTD08,FOPT10,FOPPT16}.  In particular, in the design of such
proposals, precautions should be taken since the knowledge of a non
trivial automorphism group of the public code
facilitates algebraic attacks by significantly reducing the degrees
and number of variables of the algebraic system to solve
in order to recover the secret key.

Among the recent submissions to NIST call for post quantum
cryptography, a proposal called DAGS~\cite{BBBCDGGHKONPR17} is based
on the use of quasi--dyadic (QD) generalised Srivastava codes with
extension degree $m = 2$. By {\em quasi--dyadic} we mean that the
permutation group of the code is of the form $(\Z / 2\Z)^\gamma$ for
some positive integer $\gamma$. Moreover, generalised Srivastava codes
form a proper subclass of alternant codes.  DAGS proposal takes
advantage of both usual techniques to reduce the size of the
keys. First, by using alternant codes which are close to generalised
Reed Solomon codes {\em i.e.} with an extension degree $2$. Second, by
using codes with a large permutation group.  In terms of security with
respect to key recovery attacks, DAGS parameters are chosen to be out
of reach of the algebraic attacks \cite{FOPT10,FOPPT16}. In addition,
it should be emphasised that the choice of alternant codes which are
not Goppa codes permits to be out of reach of the distinguisher by
shortening and squaring used in \cite{COT17}.

\paragraph{Our contribution}
In this article, we present an attack breaking McEliece instantiations
based on alternant codes with extension degree $2$ and a large
permutation group.  This attack permits to recover the secret key in
$O\left(n^{3+\frac{2q}{|\group|}}\right)$ operations in $\Fq$, where
$\group$ denotes the permutation group, $n$ the code length and $\F_q$
is the base field of the public code.  The key step of the attack
consists in finding some subcode of the public code referred to as
$\DC$. From this code $\DC$ and using an operation we called
{\em conductor}, the secret key can easily be recovered. For this
main step, we present two ways to proceed, the first approach is based
on a partial brute force search while the second one is based on the
resolution of a polynomial system of degree $2$.  An analysis of the
work factor of this attack using the first approach shows that DAGS
keys with respective estimated security levels 128, 192 and 256 bits
can be broken with respective approximate work factors
$2^{70}, 2^{80}$ and $2^{58}$.  For the second approach, we were not
able to provide a complexity analysis. However, its practical
implementation using Magma~\cite{BCP97} is impressively efficient on
some DAGS parameters. In particular, it permits to break claimed 256
bits security keys in less than one minute!

This attack is a novel and original manner to recover the structure of
alternant codes by jointly taking advantage of the permutation group and
the small size of the extension degree.  Even if some variant of the
attack reposes on the resolution of a polynomial system, this system has
nothing to do with those of algebraic attacks of
\cite{FOPT10,FOPPT16,FPP14}.  On the other hand, despite this attack
shares some common points with that of \cite{COT17} where the Schur product
of codes (See Section~\ref{sec:Schur} for a definition) plays a crucial role,
the keys we
break in the present article are out of reach of a distingusher by
shortening and squaring and hence our attack differs from filtration
attacks as in \cite{COT17,CMP17}.

It is worth noting that reparing DAGS scheme in order to resist to the
present attack is possible. Recently, the authors presented new
parameter sets which are out of reach of the first version of the
attack.  These new parameters are available on the current version of
the proposal\footnote{\url{https://dags-project.org/pdf/DAGS_spec.pdf}}.

%%% Debut coupe ACr
% \paragraph{Outline of the article}
% This article is organised as follows. Prerequisites on algebraic codes
% and quasi--dyadic codes are given in Section~\ref{sec:nota}.  In
% Section~\ref{sec:Schur}, we recall the definition of the Schur product
% of codes and some of its properties.  In Section~\ref{sec:conductor},
% we introduce a fundamental object for this attack called the {\em
%   conductor} of a pair of codes. Then, Section~\ref{sec:fundamental}
% presents a crucial of QD codes and their invariant subcodes which is
% crucial for the efficiency of the attack. The description of the
% attack is given in Section~\ref{sec:attack} and its complexity is
% discussed in Section~\ref{sec:complexity}.  Finally,
% Section~\ref{sec:implem} is devoted to the implementation of the
% attack and the presentation of experimental results.
%%% Fin coupe ACr

%%% Local Variables:
%%% mode: latex
%%% TeX-master: "Article"
%%% End:

\section{Notation and prerequisites}\label{sec:nota}

\subsection{Subfield subcodes and trace codes}

\begin{definition}
  Given a code $\CC$ of length $n$ over $\Fqm$, its {\em subfield
    subcode} is the subcode of vectors whose entries all lie in $\Fq$,
  that is the code:
   $$\CC \cap \Fq^n.$$
   The {\em trace code} is the image of the code by the component wise trace
   map
   $$
   \tr_{\Fqm/\Fq} (\CC) \eqdef \left\{
     \tr_{\Fqm/\Fq}(\cv) ~|~ \cv \in \CC
   \right\}.
   $$
\end{definition}

% Note that, in order to limit the amount of heavy notation, 
% we frequently omit the index ``$\Fqm/\Fq$'' of the trace in what
% follows.

Let us recall a classical and well--known result
on subfield subcodes and trace codes.

% \begin{proposition}
%   Let $\CC \subseteq \Fqm^n$ be a code of dimension $k$ and set $r = n-k$.
%   Then 
%   $$
%   \dim_{\Fq} \CC \cap \Fq^n \geq n-mr.
% %  \quad \dim_{\Fq} \tr_{\Fqm/\Fq}(\CC) \leq mr.
%   $$
% \end{proposition}

% \begin{proof}
%   \cite[Chapter 7\S 7]{MS86}. 
% \end{proof}

\begin{theorem}[Delsarte Theorem~\cite{D75}]
  \label{thm:Delsarte}
  Let $\CC  \subseteq \Fqm^n$ be a code. Then
  $$
  {(\CC \cap \Fq^n)}^\perp = \tr_{\Fqm/\Fq}(\CC^\perp).
  $$
\end{theorem}

\subsection{Generalised Reed--Solomon codes and alternant codes}
\begin{nota}
  Let $q$ be a power of prime and $k$ a positive integer. We denote by
  $\F_q[z]_{<k}$ the vector space of polynomials over $\Fq$ whose
  degree is bounded from above by $k$. Let $m$ be a positive integer,
  we will consider codes over $\Fqm$ with their subfield subcodes over
  $\Fq$. In \S~\ref{sec:Schur} and further,
  we will focus particularly on the case $m=2$.
\end{nota}

\begin{definition}[Supports and multipliers]
  A vector $\xv \in \Fqm^n$ whose entries are pairwise distinct is called
  a {\em support}. A vector $\yv \in \Fqm^n$ whose entries are all
  nonzero is referred to as a {\em multiplier}.
\end{definition}

\begin{definition}[Generalised Reed--Solomon codes]\label{def:GRS}
  Let $n$ be a posi\-tive integer, $\xv \in \Fqm^n$ be a support and
  $\yv \in \Fqm^n$ be a multiplier. The {\em generalised Reed--Solomon
    (GRS) code with support $\xv$ and multiplier $\yv$ of dimension
    $k$} is defined as
\[
\GRS{k}{\xv}{\yv} \eqdef \left\{(y_1 f(x_1), \ldots, y_n f(x_n)) ~|~
f \in \F_q[z]_{<k}\right\}.
\]
When $\yv = \onev$, the code is a {\em Reed--Solomon code} and is denoted
as $\RS{k}{\xv}$.
\end{definition}

The dual of a GRS code is a GRS code too. This is made explicit
in Lemma~\ref{lem:dual_GRS} below. Let us first introduce an
additional notation.

\begin{nota}
  Let $\xv \subseteq \Fqm^n$ be a support,
  we define the {\em polynomial} $\loc{\xv}
 \in \Fqm[z]$ as
  $$
  \loc{\xv}(z) \eqdef \prod_{i = 1}^{n} (z - x_i).
  $$
\end{nota}

\begin{lemma}\label{lem:dual_GRS}
  Let $\xv, \yv \in \Fqm^n$ be a support and a multiplier of length $n$
  and $k \leq n$. Then
  $$
  \GRS{k}{\xv}{\yv}^\perp = \GRS{n-k}{\xv}{\yv^{\perp}},
  $$
  where 
  $$
  \yv^\perp \eqdef \left(\frac{1}{\loc{\xv}'(x_1)y_1},\ldots,
  \frac{1}{\loc{\xv}'(x_n)y_n}\right),
  $$
  and $\loc{\xv}'$ denotes the derivative of the polynomial
  $\loc{\xv}$.
\end{lemma}

\begin{definition}[Alternant code]\label{def:Alternant}
  Let $m,\ n$ be positive integers such that
  $n \leq q^m$.  Let $\xv \in \Fqm^n$ be a support,
  $\yv \in \Fqm^n$ be a multiplier and $r$ be a positive integer. The
  {\em alternant code of support $\xv$, multiplier $\yv$ and degree
    $r$ over $\Fq$} is defined as
  $$
  \Alt{r}{\xv}{\yv} \eqdef \GRS{r}{\xv}{\yv}^\perp \cap \F_q^n.
  $$
  The integer $m$ is referred to as the {\em extension degree of}
  the alternant code.
\end{definition}

% \begin{proposition}\label{prop:dim_alt}
%   In the context of Definition~\ref{def:Alternant}, we have
%   $$
%   \dim \Alt{r}{\xv}{\yv} \geq n - mr.
%   $$
% \end{proposition}

% \begin{proof}
%   See \cite[Chapter 12]{MS86}. \qed
% \end{proof}

As a direct consequence of Lemma~\ref{lem:dual_GRS} and
Definition~\ref{def:Alternant}, we get the following explicit
description of an alternant code.
  \begin{equation}\label{eq:description_alternant}
  \Alt{r}{\xv}{\yv} = \cset{\frac{1}{\loc{\xv}'(x_i) y_i} f(x_i)}{i =
    1, \ldots, n}{ f \in \Fqm[z]_{< n - r}} \cap \Fq^n.
  \end{equation}
  Next, by duality and using Delsarte's Theorem
  (Theorem~\ref{thm:Delsarte}), we have
  \begin{equation}\label{eq:dual_alternant}
  \Alt{r}{\xv}{\yv}^\perp = \tr_{\Fqm/\Fq} \left(
    \cset{y_i g(x_i)}{i = 1, \ldots, n}{g \in \Fqm[z]_{<r}}
  \right).
  \end{equation}
 % where $\tr_{\Fqm/\Fq}$ denotes the component wise trace map.

  We refer the reader to \cite[Chapter 12]{MS86} for further properties
of alternant codes. Recall that the code $\Alt{r}{\xv}{\yv}$
defined in Definition~\ref{def:Alternant} has dimension $k \geq n-mr$
and equality holds in general. Moreover, these codes benefit from efficient
decoding algorithms correcting up to $\lfloor \frac r 2 \rfloor$ errors
(see \cite[Chapter 12\S 9]{MS86}).

% \subsubsection{Decoding alternant codes}
% Alternant codes come with an efficient decoding algorithm. For instance,
% see \cite[Chapter 12\S 9]{MS86}.

% \begin{fact}\label{fact:decoding}
%   Given an alternant code $\Alt{r}{\xv}{\yv}$,
%   there exists an efficient decoding algorithm correcting up
%   to $t = \lfloor \frac r 2 \rfloor$ errors.
%   This decoding algorithm can be built from the knowledge of 
%   the pair $(\xv, \yv)$.
% \end{fact}

\subsubsection{Fully non degenerate alternant codes}
We conclude this subsection on alternant codes by a definition
which is useful in the sequel. 
\begin{definition}\label{def:full_ndgn}
An alternant code $\Alt{r}{\xv}{\yv}$ is said to be {\em
fully non degenerate} if
it satisfies the two following conditions.
\begin{enumerate}[(i)]
\item A generator matrix of $\Alt{r}{\xv}{\yv}$ has no zero column.
\item $\Alt{r}{\xv}{\yv} \neq \Alt{r+1}{\xv}{\yv}$.
\end{enumerate}
\end{definition}

%It should be emphasised that,
Most of the time, an alternant code is fully non degenerate.

\subsection{Punctured and shortened codes}\label{subsec:shortening}

The notions of \textit{puncturing} and \textit{shortening} are
classical ways to build new codes from existing
ones. % These notions will be useful for the attack.
We recall here their definition.

\begin{definition}
  Let $\CC$ be a code of length $n$ and $\Ind \subseteq \{1, \ldots, n\}$.
  The {\em puncturing} and the {\em shortening of $\CC$ at $\Ind$} are
  respectively defined as the codes
  \begin{align*}
  \pu{\CC}{\Ind} &\eqdef \{(c_i)_{i \in \{1, \ldots, n\}\backslash\Ind} ~|~
  \cv \in \CC\},\\
  \sh{\CC}{\Ind} &\eqdef \{(c_i)_{i \in \{1, \ldots, n\}\backslash\Ind} ~|~
  \cv \in \CC\ {\rm such\ that}\ \forall i \in \Ind,\ c_i = 0\}.
  \end{align*}
\end{definition}

% \begin{definition}
%   Let $\CC$ be a code of length $n$ and $\Ind \subseteq \{1, \ldots, n\}$.
%   The {\em shortening of $\CC$ at $\Ind$} is defined as the code
%   \[
%   \sh{\CC}{\Ind} \eqdef \pu{\{\cv \in \CC ~|~
%   \forall i \in \Ind,\ c_i = 0\}}{\Ind}.
%   \]
% \end{definition}

% \begin{remark}
%   The shortening of a code at a subset of positions $\Ind$ is the
%   subcode of vectors whose components with indexes in $\Ind$ are all
%   zero and these zero components have been removed.  It is sometimes
%   useful when shortening a code, {\bf not} to remove these zero
%   components in order to keep a code of the same length. It is for
%   instance necessary when one wants to compute the Schur product of a
%   code and one of its shortenings. To limit the amount of notation,
%   we denote both codes by $\sh{\CC}{\Ind}$ : the one whose zero
%   components have not been removed and the one where they have been
%   removed.
% \end{remark}

Let us finish by recalling the following classical result.

\begin{nota}\label{nota:short}
  Let $\xv \in \Fqm^n$ be a vector and $\Ind \subseteq \{1, \ldots, n\}$.
  Then, the vector $\xv_{\Ind}$ denotes the vector obtained from
  $\xv$ be removing the entries whose indexes are in $\Ind$. 
\end{nota}

\begin{proposition}\label{prop:short_alt}
  Let $m, r$ be positive integers.
  Let $\xv, \yv\in \F_{q^m}^n$ be as in Definition~\ref{def:Alternant}.
  Let $\Ind \subseteq \{1, \ldots, n\}$. Then
  $$
  \sh{\Alt{r}{\xv}{\yv}}{\Ind} = \Alt{r}{\xv_{\Ind}}{\yv_{\Ind}}.
  $$
\end{proposition}

\begin{proof}
  See for instance \cite[Proposition 9]{COT17}. \qed
\end{proof}

\subsection{Quasi--dyadic codes, quasi-dyadic alternant codes}

Quasi--dyadic (QD) codes are codes with a nontrivial permutation group
isomorphic to $(\Z/2\Z)^\gamma$ for some positive integer
$\gamma$. Such a code has length $n = 2^\gamma n_0$.
The permutation group of the code is composed of permutations,
each one being a product of transpositions with disjoint supports.
The example of interest in the present article is the case of
QD--alternant codes. In what follows, we explain how to create them.

\begin{nota}
From now on, $q$ denotes a power of $2$ and $\ell$ denotes the
positive integer such that $q = 2^\ell$.
\end{nota}
% Since
% our attack focuses on QD--alternant codes over a quadratic extension
% and for the sake of simplicity, we limit our description to the case of
% QD--alternant codes with extension degree $m = 2$. The construction
% below generalizes easily to an arbitrary extension degree.

\begin{itemize}
\item Let $\group \subset \Fqm$ be an additive subgroup with $\gamma$ generators,
  i.e. $\group$ is an $\F_2$--vector subspace of $\Fqm$ of dimension $\gamma$
  with an $\F_2$--basis $a_1, \ldots, a_\gamma$.
  Clearly, as an additive group, $\group$ is isomorphic to $(\Z/2\Z)^\gamma$.
  The group $\group$ acts on $\Fqm$ by translation:
  for any  $a \in \group$, we denote by $\tau_a$ the translation 
  $$
  \tau_a : \map{\Fqm}{\Fqm}{x}{x+a}.
  $$
\item Using the basis $(a_1, \ldots, a_\gamma)$, we fix an ordering in $\group$
  as follows. Any element $u_1 a_1 + \cdots + u_\gamma a_\gamma\in \group$ can
  be regarded as an element $(u_1, \ldots, u_\gamma)\in(\Z/2\Z)^\gamma$
  and we sort them by lexicographic order. For instance, if $\gamma=3$:
  $$
  0 < a_1 < a_2 < a_1 + a_2 < a_3 < a_1 + a_3 < a_2 + a_3 < 
  a_1 + a_2 + a_3.
  $$
\item Let $n = 2^\gamma n_0$ for some positive $n_0$ and such that
  $n \leq q^m$. Let $\xv \in \Fqm^n$ be a support which
  splits into $n_0$ blocks of $2^\gamma$ elements of $\Fqm$, each
  block being an orbit under the action of $\group$ by translation on
  $\Fqm$ sorted using
  the previously described ordering.
  For instance, suppose $\gamma = 2$, then such an $\xv$
  is of the form,
  \begin{equation}\label{eq:x_example}
    \begin{array}{rl}
  \xv  = & (t_1, t_1+a_1, t_1+a_2, t_1+a_1+a_2, \ldots,\\
  & \qquad \ldots, t_{n_0}, t_{n_0}+a_1, t_{n_0}+a_2, t_{n_0}+a_1+a_2), 
    \end{array}
  \end{equation}
  where the $t_i$'s are chosen to have disjoint orbits
  under the action of $\group$ by translation on $\Fqm$.
\item Let $\yv \in \Fqm^n$ be a multiplier which also splits into $n_0$
  blocks of length $2^\gamma$ whose entries are equal.
\item Let $r$ be a positive integer and consider the code
  $\Alt{r}{\xv}{\yv}$.
\item The set of entries of $\xv$ is globally invariant under the
  action of $\group$ by translation. In particular, for any
  $a \in \group$, the translation $\tau_a$ induces a permutation of
  the code $\Alt{r}{\xv}{\yv}$.
  %This permutation is a product of
  %$\frac n 2$ transpositions with disjoint supports.
  We refer this permutation to as $\sigma_a$.
  For instance,
  reconsidering Example~(\ref{eq:x_example}), the permutations
  $\sigma_{a_1}$ and $\sigma_{a_1+a_2}$ are respectively of the form
  \begin{align*}
  \sigma_{a_1} & = (1,2)(3,4) \cdots (n-3, n-2)(n-1, n)\\
  \sigma_{a_1+a_2} &= (1, 4)(2, 3) \cdots (n-3, n)(n-2, n-1).
  \end{align*}
  The group of permutations $\{\sigma_a ~|~ a\in \group\}$ is
  isomorphic to $\group$ and hence to $(\Z/2\Z)^\gamma$.
  % Since it is isomorphic to $\group$ and in order to
  % limit the amount of notation,
  For convenience,
  we also denote this group
  of permutations by $\group$.
\end{itemize}

\begin{proposition}
  For any $r >0$, the code $\Alt{r}{\xv}{\yv}$
  is quasi--dyadic.
\end{proposition}

\begin{proof}
  See for instance \cite[Chapter 5]{P15}. \qed
\end{proof}

\subsection{Invariant subcode of a quasi--dyadic
  code}\label{ss:inv_alternant}

\begin{definition}\label{def:invariant_code}
  Given a code $\CC$ with a non--trivial permutation group $\group$,
  we define the code $\inv{\CC}{\group}$ as the subcode
  of $\CC$:
$$
\inv{\CC}{\group} \eqdef \{\cv \in \CC ~|~ \forall \sigma \in \group,\
\sigma(\cv) = \cv\}.
$$
\end{definition}

The invariant subcode has repeated entries since on any orbit of the
support under the action of $\group$, the entries of a codeword are
equal.  This motivates an alternative definition of the invariant code
where repetitions have been removed.

\begin{definition}\label{def:punct_inv}
  In the context of Definition~\ref{def:invariant_code},
  let $\cv \in \Fqm^n$ be a vector such that
  for any $\sigma \in \group$, $\sigma (\cv) = \cv$.
  We denote by $\overline{\cv}$ the vector
  obtained by keeping only one entry per orbit under the
  action of $\group$ on the support.
  We define the {\em
    invariant code with non repeated entries}
  as
  $$
  \punctinv{\CC}{\group} \eqdef \left\{
   { \overline{\cv} ~|~ \cv \in  \inv{\CC}{\group}}
  \right\}.
  $$
\end{definition}

We are interested in the structure of invariant of QD alternant codes.
To study this structure, we first need to recall some basic notions
of additive polynomials.

\subsubsection{Additive polynomials}

\begin{definition}
  An {\em additive polynomial} $P\in \Fqm[z]$ is a polynomial whose
  monomials are all of the form $z^{2^i}$ for $i \geq 0$.
  Such a polynomial satisfies
  $P(a+b) = P(a)+ P(b)$ for any $a, b \in \Fqm$.
\end{definition}

The zero locus of an additive polynomial in $\Fqm$ is an additive
subgroup of $\Fqm$ and such polynomials satisfy some interpolation
properties.

\begin{proposition}
  Let $\group \subset \F_{q^m}$ be an additive group of cardinality
  $2^\gamma$.  There exists a unique additive polynomial
  $\psi_\group \in \Fqm[z]$ which is monic of degree $2^\gamma$ and
  vanishes at any element of $\group$.
\end{proposition}

\begin{proof}
  %%% Begin cut ACr
%   Let $a_1, \ldots, a_\gamma$ be a set of
%   generators of $\group$.  The polynomial $\psi_\group$ can be
%   constructed using the so--called {\em Moore determinant}:
% $$
% \psi_\group(z) \eqdef
% {\left|
%   \begin{array}{ccc}
%     a_1 & \cdots & a_{\gamma} \\
%     \vdots &       & \vdots \\
%     a_1^{2^{{\gamma}-1}} & \cdots & a_{\gamma}^{2^{{\gamma}-1}}
%   \end{array}
% \right|}^{-1}\cdot
% \left|
%   \begin{array}{cccc}
%     a_1 & \cdots & a_{\gamma} & z \\
%     a_1^2 & \cdots & a_{\gamma}^2 & z^2 \\
%     \vdots &       & \vdots & \vdots \\
%     a_1^{2^{\gamma}} & \cdots & a_{\gamma}^{2^{\gamma}} & z^{2^{\gamma}}
%   \end{array}
% \right|.
% $$
% See \cite[Proposition 1.3.5]{G96b} for further details.
%%% End cut ACr
See \cite[Proposition 1.3.5 \& Lemma 1.3.6]{G96b}.
\qed
\end{proof}

\begin{nota}\label{nota:psi_G}
  From now on, given an additive subgroup $\group \subseteq \Fqm$,
  we always denote by $\psi_\group$ the unique monic additive polynomial
  of degree $|\group|$ in $\Fqm[z]$ that vanishes on $\group$.
\end{nota}

\subsubsection{Invariant of a quasi--dyadic alternant code}
It turns out that the invariant code with non repeated entries of a QD
alternant code is an alternant
code too.  This relies on the following classical result of invariant
theory for which a simple proof can be found in \cite{FOPPT16}.

\begin{theorem}\label{thm:invariant_poly}
  Let $f \in \Fqm[z]$ and $\group \subset \Fqm$ be an additive
  subgroup.  Suppose that for any $a \in \group$, $f(z) =
  f(z+a)$.
  Then, there exists $h \in \Fqm[z]$ such that
  $f(z) = h(\psi_\group(z))$, where $\psi_\group$ is the monic
  additive polynomial of degree $|\group|$ vanishing at any element of
  $\group$.
\end{theorem}

This entails the following result on the structure of the invariant code
of an alternant code. We refer to Definition~\ref{def:punct_inv}
for the notation in the following statement.

\begin{theorem}\label{thm:invariant_alternant}
  Let $\CC = \Alt{r}{\xv}{\yv}$ be a $QD$--alternant code with
  permutation group $\group$ of order $2^\gamma$.
  Set $r' = \left\lfloor \frac{r}{2^{\gamma}} \right\rfloor$.
  Then,
  $$
  \punctinv{\CC}{\group} = \Alt{r'}{\overline{\psi_\group(\xv)}}{
    \overline{\yv}},
  $$
\end{theorem}

\begin{proof}
  See \cite{B17}.\qed 
\end{proof}

\subsection{DAGS}\label{ss:DAGS}
Among the schemes recently submitted to NIST, the submission
DAGS~\cite{BBBCDGGHKONPR17} uses as a primitive a McEliece encryption
scheme based on QD generalised Srivastava codes.  It is well known
that generalised Srivastava codes form a subclass of alternant codes
\cite[Chapter 12]{MS86}. Therefore, this proposal lies in the scope of
the attack presented in what follows.

Parameters proposed in DAGS submission are listed in
Table~\ref{tab:Parameters_DAGS}.
\begin{table}[!h]
\renewcommand\arraystretch{1.3}
\addtolength{\tabcolsep}{.1cm}
  \centering
  \begin{tabular}{|c||c|c|c|c|c|c|c|c|}
    \hline
    Name & $q$ & $m$ & $n$ & $n_0$ & $k$ & $k_0$ & $\gamma$ &
         $r_0$ \\
    \hline \hline
    \dags{1} & $2^5$ & $2$ & $832$ & $52$ & $416$ & $26$ & $4$ & $13$ \\
    \hline
    \dags{3} & $2^6$ & 2 & 1216 & 38 & 512 & 16 & 5 & 11 \\
    \hline
    \dags{5} & $2^6$ & 2 & 2112 & 33 & 704 & 11 & 6 & 11 \\
    \hline 
  \end{tabular}
  \vspace{.2cm}
  \caption{Parameters proposed in DAGS.}
  \label{tab:Parameters_DAGS}
\renewcommand\arraystretch{1}
\end{table}

\noindent
Let us recall what do the parameters
$q, m, n, n_0, k, k_0, \gamma, r_0$ stand for:
\begin{itemize}
\item $q$ denotes the size of the base field of the alternant code;
\item $m$ denotes the extension degree. Hence the GRS code above the
  alternant code is defined over $\Fqm$;
\item $n$ denotes the length of the QD alternant code;
\item $n_0$ denotes the length of the invariant code with non repeated
  entries
  $\punctinv{\Alt{r}{\xv}{\yv}}{\group}$, where $\group$ denotes the 
  permutation group.
\item $k$ denotes the dimension of the QD alternant code;
\item $k_0$ denotes the dimension of the  invariant code;
\item $\gamma$ denotes the number of generators of $\group$, {\em i.e.}
  $\group \simeq (\Z/2\Z)^{\gamma}$;
\item $r_0$ denotes the degree of the  invariant code with non repeated
  entries, which is 
  alternant according to Theorem~\ref{thm:invariant_alternant}.
\end{itemize}

\begin{remark}
  The indexes ${\tt 1}, {\tt 3}$ and ${\tt 5}$ in the parameters names
  correspond to security levels according to NIST's call. Level 1, corresponds
  to 128 bits security with a classical computer, Level 3 to 192 bits security
  and Level 5 to 256 bits security.
\end{remark}

% \ac{Ce qui suit doit sans doute être repris plus tard avec soin.
%   Pour le moment, ça ressemble fortement à un bâton bien tendu pour se
%   faire taper dessus.}
In addition to the set of parameters of
Table~\ref{tab:Parameters_DAGS}, we introduce self chosen smaller
parameters listed in Table~\ref{tab:Parameters_BABYDAGS}. They {\bf do
  not} correspond to claimed secure instantiations of the scheme but
permitted to test some of our assumptions by computer aided
calculations.% which would be out of reach of our computing
% capabilities when tested on regular DAGS parameters as in
% Table~\ref{tab:Parameters_DAGS}.  For instance, when we argue 
% further that \dags 3, can be broken with an approximate work factor
% of $2^{90}$ operations, such an attack is out of reach of our computing
% capacities and hence we simulated this attack on small scale parameters
% to give evidences that the attack would work properly with the 
% estimated work factor.

\begin{table}[h]
  \centering
  \renewcommand\arraystretch{1.3}
  \addtolength{\tabcolsep}{.1cm}
  \begin{tabular}{|c||c|c|c|c|c|c|c|c|}
    \hline
    Name & $q$ & $m$ & $n$ & $n_0$ & $k$ & $k_0$ & $\gamma$ &
         $r_0$ \\
    \hline \hline
    \dags{0} & $2^4$ & $2$ & $240$ & $15$ & $80$ & $5$ & $4$ & $5$ \\
    \hline
  \end{tabular}
  \vspace{.2cm}
  \caption{Small scale parameters, {\bf not} proposed in DAGS.}
  \label{tab:Parameters_BABYDAGS}
  \renewcommand\arraystretch{1}
\end{table}
%\ac{Elise, tu pourras me confirmer les param\`etres ci-dessus stp?}

%%% Local Variables:
%%% mode: latex
%%% TeX-master: "Article"
%%% End:

\section{Schur products}\label{sec:Schur}

From now on and unless otherwise specified, the extension degree 
$m$ will be equal to $2$. This is the context of any proposed
parameters in DAGS.

\subsection{Product of vectors}
  The component wise product of two vectors in $\Fq^n$ is denoted by
\[
\av \star \bv \eqdef (a_1b_1, \ldots, a_n b_n).
\]
Next, for any positive integer $t$ we define $\av^{\star t}$ as 
$$
\av^{\star t} \eqdef
\underbrace{\av \star \cdots \star \av}_{t\ \textrm{times}}.
$$
More generally, given a polynomial $P\in \F_q[z]$ we define $P(\av)$
as the vector $(P(a_1), \ldots, P(a_n))$. In particular, given
$\av \in \F_{q^2}^n$, we denote by $\tr (\av)$ and $\nr(\av)$
the vectors obtained by
applying respectively the trace and the norm map component by component:
\begin{align*}
\tr(\av) &\eqdef (a_1 + a_1^q , \ldots, a_n + a_n^q)\\
\nr(\av) & \eqdef (a_1^{q+1}, \ldots, a_n^{q+1}).
\end{align*}
Finally, the all one vector $(1, \ldots, 1)$, which is the unit vector of
the algebra $\F_q^n$ with operations $+$ and $\star$ is denoted by $\onev$.

\subsection{Schur product of codes}
The {\em Schur product} of two codes $\code A$ and $\code B \subseteq \F_q^n$
is defined as
\[
\code A \star \code B \eqdef  {\left\langle
  \av \star \bv ~|~ \av \in \code A, \ \bv \in \code B
  \right\rangle}_{\Fq}.
\]
In particular, $\sq{\code A}$ denotes the \emph{square code} of a code $\code A$: $\sq{\code A}\eqdef \code A \star \code A$.

\subsection{Schur products of GRS and alternant codes}
The behaviour of GRS and of some alternant codes with respect to
the Schur product is very different from that of random codes.
This provides a manner to distinguish GRS codes from random ones
and leads to a cryptanalysis of GRS based encryption schemes
\cite{W10,CGGOT14,COTG15}. Some alternant codes, namely
Wild Goppa codes with extension degree 2 have been also subject
to a cryptanalysis based on Schur products computations \cite{COT14,COT17}.

Here we recall an elementary but crucial result.

\begin{theorem}\label{thm:prod_GRS}
  Let $\xv \in \Fqm^n$ be a support and
  $\yv, \yv' \in \Fqm^n$ be multipliers.
  Let $k, k'$ be two positive integers, then
  $$
  \GRS{k}{\xv}{\yv} \star \GRS{k'}{\xv}{\yv'} = 
  \GRS{k+k'-1}{\xv}{\yv \star \yv'}.
  $$
\end{theorem}

\begin{proof}
  See for instance \cite[Proposition 6]{CGGOT14}.\qed
\end{proof}

%\ac{A voir si on a la place et si on en a vraiment besoin, on citera CCMZ ici.}

%%% Local Variables:
%%% mode: latex
%%% TeX-master: "Article"
%%% End:

\section{Conductors} \label{sec:conductor}
In this section, we introduce a fundamental object in the
attack to follow. This object was already used in
\cite{CMP17,COT17} without being named. We chose here to 
call it {\em conductor}. The rationale behind this terminology
is explained in Remark~\ref{rem:conductor}.

\begin{definition}\label{def:conductor}
  Let $\CC$ and $\DC$ be two codes of length $n$ over $\F_q$.
  The {\em conductor of $\DC$ into $\CC$} is defined as the
  largest code $\ZC \subseteq \F_q^n$ such that $\DC \star \ZC \subseteq \CC$.
  That is:
  $$
  \cond{\DC}{\CC} \eqdef \{\uv \in \F_q^n ~|~ \uv \star \DC \subseteq \CC\}.
  $$
\end{definition}

\begin{proposition}\label{prop:formula_conductor}
  Let $\DC, \CC \subseteq \F_q^n$ be two codes, then
  $$
  \cond{\DC}{\CC} = {\left( \spc{\DC}{\CC^\perp} \right)}^\perp.
  $$
\end{proposition}

\begin{proof}
  See \cite{CMP17,COT17}. \qed
\end{proof}

\begin{remark}\label{rem:conductor}
  The terminology {\em conductor} has been borrowed from number theory
  in which the conductor of two subrings $\mathcal O, \mathcal O'$ of
  the ring of integers $\mathcal O_K$ of a number field $K$ is the
  largest ideal $\mathfrak P$ of $\mathcal O_K$ such that
  $\mathfrak P \cdot \mathcal O \subseteq \mathcal O'$.
\end{remark}

\subsection{Conductors of GRS codes}\label{ss:conductor_GRS}

\begin{proposition}\label{prop:cond_GRS}
  Let $\xv, \yv \in \Fqm^n$ be a support and a multiplier.
  Let $k \leq k'$ be two integers less than $n$. Then,
  $$
  \cond{\GRS{k}{\xv}{\yv}}{\GRS{k'}{\xv}{\yv}}
  = \RS{k'-k+1}{\xv}.
  $$
\end{proposition}

\begin{proof}
  Let $\EC$ denote the conductor.
  From Proposition~\ref{prop:formula_conductor} and
  Lemma~\ref{lem:dual_GRS}, 
  $$
  \EC^\bot = \spc{\GRS{k}{\xv}{\yv}}{\GRS{n-k'}{\xv}{\yv^\bot}}
   = \GRS{n-k'+k-1}{\xv}{\spc{\yv}{\yv^\bot}}.
  $$
  Note that
  $$
  \spc{\yv}{\yv^\bot} = \left( \frac{1}{\loc{\xv}' (x_1)}, \ldots, 
  \frac{1}{\loc{\xv}'(x_n)}\right).
  $$
  Then, using Lemma~\ref{lem:dual_GRS} again, we get
  $$
  \EC = \GRS{k'-k+1}{\xv}{{(\spc{\yv}{\yv^\bot})}^\bot} = \RS{k'-k+1}{\xv}.
  $$\qed
\end{proof}

Let us emphasize a very interesting aspect of
Proposition~\ref{prop:formula_conductor}. We considered the conductor 
of a GRS code into another one having the same support and multiplier.
The point is that the conductor {\bf does not depend on $\yv$}. Hence
the computation of a conductor permits to get rid of the multiplier
and to obtain a much easier code to study: a Reed--Solomon code.

\subsection{An illustrative example : recovering the structure
  of GRS codes using conductors}\label{ss:illustrative_GRS}

Before presenting the attack on QD--alternant codes, we
propose first to describe
a manner to recover the structure of a GRS code. This may help
the reader to understand the spirit the attack to follow.

Suppose we know
a generator matrix of a code
$\CC_k = \GRS{k}{\xv}{\yv}$ where $(\xv, \yv)$ are unknown.
In addition, suppose that we know
a generator matrix of the subcode $\CC_{k-1} = \GRS{k-1}{\xv}{\yv}$
which has codimension $1$ in $\CC_k$.
% We aim at finding a pair $(\xv', \yv')$ such that $\CC_k =
% \GRS{k}{\xv'}{\yv'}$ and we can proceed as follows.
First compute the conductor
$$
\XC = \cond{\CC_{k-1}}{\CC_k}.
$$
From Proposition~\ref{prop:cond_GRS}, the conductor $\XC$ equals
$\RS{2}{\xv}$. This code has dimension $2$ and is spanned by
$\onev$ and $\xv$.
We claim that, from the knowledge of $\XC$, a pair $(\xv', \yv')$
such that $\CC_k = \GRS{k}{\xv'}{\yv'}$
can be found easily by using techniques which are very similar
from those presented further in \S~\ref{ss:finishing}.

Of course, there is no reason that we could know both
$\GRS{k}{\xv}{\yv}$ and $\GRS{k-1}{\xv}{\yv}$. However, we will see
further that the quasi--dyadic structure permits to find interesting
subcodes whose conductor may reveal the secret structure of the code.

\subsection{Conductors of alternant codes}

When dealing with alternant codes, having an exact description
of the conductors like in Proposition~\ref{prop:cond_GRS} becomes
difficult. We can at least prove the following theorem.

\begin{proposition}\label{prop:Conductor_alternant}
  Let $\xv, \yv \in \Fqq^n$ be a support and a multiplier. Let $r'\geq r$
  be two positive integers.
  Then,
  \begin{equation}\label{eq:Conductor_alternant}
  \RS{r'-r+1}{\xv} \cap \Fq^n \subseteq
  \cond{\Alt{r'}{\xv}{\yv}}{\Alt{r}{\xv}{\yv}}.
  \end{equation}
\end{proposition}

\begin{proof}
Consider the Schur product
\begin{align*}
  \spc{\left(\RS{r'-r+1}{\xv} \cap \Fq^n \right)}{&\Alt{r'}{\xv}{\yv}} \\
    & = \spc{\left(\RS{r'-r+1}{\xv} \cap \Fq^n \right)}
            {(\GRS{n-r'}{\xv}{\yv^\perp} \cap \Fq^n)}\\
  & \subseteq (\spc{\RS{r'-r+1}{\xv}}{\GRS{n-r'}{\xv}{\yv^\perp}})
    \cap \Fq^n.
\end{align*}
Next, using Theorem~\ref{thm:prod_GRS},
\begin{align*}
  \spc{\left(\RS{r'-r+1}{\xv} \cap \Fq^n \right)}{\Alt{r'}{\xv}{\yv}}
  & \subseteq \GRS{n-r}{\xv}{\yv^\perp} \cap \Fq^n\\
  & \subseteq \Alt{r}{\xv}{\yv}.
\end{align*}
The last inclusion is a consequence of
Lemma~\ref{lem:dual_GRS} and Definition~\ref{def:Alternant}.\qed
\end{proof}

\subsection{Why the straigthforward generalisation of
the illustrative example fails for alternant codes}
Compared to Proposition~\ref{prop:cond_GRS},
Proposition~\ref{prop:Conductor_alternant} provides only an
inclusion. However, it turns out that we experimentally observed that
the equality frequently holds.

On the other hand, even if inclusion
(\ref{eq:Conductor_alternant}) was an equality, the attack described
in \S~\ref{ss:illustrative_GRS} could not be straightforwardly generalised to
alternant codes. Indeed, suppose we know two alternant codes with
consecutive degrees $\Alt{r+1}{\xv}{\yv}$ and $\Alt{r}{\xv}{\yv}$.
Then, Proposition~\ref{prop:Conductor_alternant} would yield
\begin{equation}\label{eq:illustrative_alternant}
\RS{2}{\xv} \cap \Fq^n \subseteq \cond{\Alt{r+1}{\xv}{\yv}}{
\Alt{r}{\xv}{\yv}}.
\end{equation}
Suppose that the above inclusion is actually an equality; as we just said
this is in general what happens.
The point is that as soon as $\xv$ has one entry in $\Fqq \setminus \Fq$,
then $\RS{2}{\xv} \cap \Fq^n$ is reduced to the code spanned by $\onev$
and hence cannot provide any relevant information.

The previous discussion shows that, if we want to generalise the
toy attack described in \S \ref{ss:illustrative_GRS} to alternant codes,
we cannot use
a pair of alternant codes with consecutive degrees. In light
of Proposition~\ref{prop:Conductor_alternant}, the gap between the degrees
$r$ and $r'$ of the two alternant codes should be large enough to provide
a non trivial conductor. A sufficient condition for this is that
$\RS{r'-r+1}{\xv} \cap \Fq^n$ is non trivial. This motivates the introduction
of a code we called the {\em norm trace code}.

\subsection{The norm--trace code}

\begin{nota}
  In what follows, we fix $\alpha \in \Fqq$ such that $\tr (\alpha) = 1$.
  In particular, $(1, \alpha)$ forms an $\Fq$--basis of $\Fqq$.
\end{nota}

\begin{definition}[Norm trace code]
  Let $\xv \in \Fqq^n$ be a support.
  The {\em norm--trace code} $\NTC(\xv) \subseteq \Fq^n$ is defined
  as 
  $$
  \NTC (\xv) \eqdef \langle \onev, \tr(\xv), \tr(\alpha \xv), \nr(\xv)
  \rangle_{\Fq}.
  $$
\end{definition}

This {\em norm trace code} turns out to be the code we will extract
from the public key by conductor computations. To relate it with
the previous discussions, we have the following statement
whose proof is straightforward.

\begin{proposition}
  \label{prop:NT_subset_SRS}
    Let $\xv \in \Fqq^n$  be a support.
  Then, for any $k > q+1$,
  we have
  \begin{equation}\label{eq:NT_subset_SRS}
  \NTC (\xv) \subseteq \RS{k}{\xv} \cap \Fq^n.
  \end{equation}
\end{proposition}

\begin{remark}\label{rem:experiment_NT}
It addition to this statement, we observed experimentally that for
$2q+1 > k > q+1$ inclusion~(\ref{eq:NT_subset_SRS}) is in general
an equality.  
\end{remark}

\subsection{Summary and a heuristic}
First, let us summarise the previous discussions.
\begin{itemize}
\item If we know a pair of alternant codes $\Alt{r}{\xv}{\yv}$
  and $\Alt{r'}{\xv}{\yv}$ such that $q < r'-r$, then
  $\cond{\Alt{r'}{\xv}{\yv}}{\Alt{r}{\xv}{\yv}}$ is non trivial
  since, according to Propositions~\ref{prop:Conductor_alternant}
  and to~(\ref{eq:NT_subset_SRS}), it contains the norm--trace code.
\item Experimentally, we observed that if $q < r' - r < 2q$, then,
  almost every time, we have
  $$\cond{\Alt{r'}{\xv}{\yv}}{\Alt{r}{\xv}{\yv}} = \NTC (\xv).$$
\item One problem remains: given an alternant code
  $\Alt{r}{\xv}{\yv}$, how to get a subcode $\Alt{r'}{\xv}{\yv}$ in
  order to apply the previous results? This will be explained in
  \S~\ref{sec:fundamental} and \ref{sec:attack} in which we show that
  for quasi--dyadic alternant codes it is possible to get a subcode
  $\DC \subseteq \Alt{r}{\xv}{\yv}$ such that
  $\DC \subseteq \Alt{r'}{\xv}{\yv}$ for some $r'$ satisfying
  $r'-r > q+1$.

  Moreover, it turns out that
  $\Alt{r'}{\xv}{\yv}$ can be replaced by a subcode without changing
  the result of the previous discussions. This is what is argued in
  the following heuristic.
\end{itemize}

\begin{heuristic}\label{heur:strong}
  In the context of Proposition~\ref{prop:Conductor_alternant},
  suppose that $q < r-r' < 2q$.
  Let $\DC$ be a subcode of
  $\Alt{r'}{\xv}{\yv}$ such that
  \begin{enumerate}[(i)]
    \item\label{item:(i)} $\dim \DC \cdot \dim \Alt{r}{\xv}{\yv}^\perp \geq n$;
    \item\label{item:(ii)} $\DC \not \subset \Alt{r'+1}{\xv}{\yv}$;
    \item\label{item:(iii)} a generator matrix of $\DC$ has no zero column.
  \end{enumerate}
  Then, with a high probability,
  $$
  \cond{\DC}{\Alt{r}{\xv}{\yv}} = \NTC (\xv).
  $$
\end{heuristic}

\noindent Let us give some evidences for this heuristic.
% Since
% Heuristic~\ref{heur:strong} is the strongest one and clearly entails
% Heuristic~\ref{heur:weak}, let us discuss it.
From Proposition~\ref{prop:formula_conductor}, 
$$
\cond{\DC}{\Alt{r}{\xv}{\yv}} = {\left(
    \spc{\DC}{\Alt{r}{\xv}{\yv}^\perp} \right)}^\perp.
$$
From~(\ref{eq:dual_alternant}), we have
%\begin{align*}
$\Alt{r}{\xv}{\yv}^\perp = \tr_{\Fqq/\Fq}
(\GRS{r}{\xv}{\yv}).
$
%\end{align*}
Since $\DC$ is a code over $\Fq$ and by the $\Fq$--linearity of the
trace map, we get
$$
\spc{\DC}{\Alt{r}{\xv}{\yv}^\perp} =
\tr_{\Fqq/\Fq}
\left(\spc{\DC}{\GRS{r}{\xv}{\yv}}\right).
$$
Since $\DC \subseteq \Alt{r'}{\xv}{\yv}$
then, from (\ref{eq:description_alternant}), it is a subset of
a GRS code. Namely,
$$
\DC \subseteq \GRS{n-r'}{\xv}{\yv^\perp},\quad {\rm where}\quad \yv^\perp =
\left(\frac{1}{\loc{\xv}'(x_1)y_1}, \ldots, \frac{1}{\loc{\xv}'(x_n)y_n}
\right).
$$
Therefore, thanks to Theorem~\ref{thm:prod_GRS}, we get
\begin{equation}\label{eq:D*Alt^perp}
\spc{\DC}{\Alt{r}{\xv}{\yv}^\perp} \subseteq
\tr_{\Fqq/\Fq} \left(\GRS{n-r'+r-1}{\xv}{\spc{\yv}{\yv^\perp}}\right).
\end{equation}
Note that $\spc{\DC}{\Alt{r}{\xv}{\yv}^\perp}$ is spanned
by $\dim \DC \cdot \dim \Alt{r}{\xv}{\yv}^\perp$ generators which are
obtained by computing the Schur products of elements of a basis of
$\DC$ by elements of a basis of $\Alt{r}{\xv}{\yv}^\perp$. By
(\ref{item:(i)}), the number of such generators exceeds $n$.  For this
reason, it is reasonable to hope that this Schur product
fills in the target code and that,
$$
\spc{\DC}{\Alt{r}{\xv}{\yv}^\perp} = \tr_{\Fqq/\Fq} \left(
  \GRS{n-r'+r-1}{\xv}{\spc{\yv}{\yv^\perp}} \right).
$$
Next, we have
$$
\spc{\yv}{\yv^\perp} = \left( \frac{1}{\loc{\xv}'(x_1)}, \ldots,
  \frac{1}{\loc{\xv}'(x_n)} \right).
$$
Therefore, using Lemma~\ref{lem:dual_GRS}, we conclude that
$$
\left({\spc{\DC}{\Alt{r}{\xv}{\yv}^\perp}}\right)^\perp = \RS{r'-r+1}{\xv} \cap \Fq^n.
$$
Using Remark~\ref{rem:experiment_NT}, we get the result.

\begin{remark}
  Assumption~(\ref{item:(ii)}) permits to avoid the situation where
  the conductor could be the subfield subcode of a larger
  Reed--Solomon code. Assumption~(\ref{item:(iii)}) permits to avoid the
  presence of words of weight $1$ in the conductor that would not be
  elements of a Reed--Solomon code.
\end{remark}

\paragraph{Further discussion on the Heuristic}
In all our computer experiments, we never observed any phenomenon
contradicting this heuristic.

\section{Fundamental degree properties
  of the invariant subcode
  of a QD alternant code}\label{sec:fundamental}
A crucial statement for the attack is:

\begin{theorem}\label{thm:deg_inv}
  Let $\xv, \yv \in \Fqq^n$ be a support and a multiplier.
  Let $s$ be an integer of the form $s = 2^\gamma s_0$.
  Suppose that $\Alt{s_0}{\overline{\psi_{\group}(\xv)}}{\overline{\yv}}$
  is fully non degenerate (see Definition~\ref{def:full_ndgn} and
  \S~\ref{ss:inv_alternant} for notation $\psi_\group$ and
  $\overline {\yv}$). Then,
  \begin{enumerate}[(a)]
  \item\label{item:a}
    $\inv{\Alt{s}{\xv}{\yv}}{\group} \subseteq \Alt{s + |\group| -
      1}{\xv}{\yv};$
  \item\label{item:b}
    $\inv{\Alt{s}{\xv}{\yv}}{\group} \not \subseteq \Alt{s +
      |\group|}{\xv}{\yv}.$
  \end{enumerate}
\end{theorem}

%%% Begin cut ACr
% \begin{remark}
%   Note that the sequence of alternant codes for increasing degrees
%   is decreasing:
%   $$
%   \Alt{s}{\xv}{\yv} \supseteq \Alt{s+1}{\xv}{\yv} \supseteq \cdots
%   \supseteq \Alt{s+|\group|-1}{\xv}{\yv} \supseteq
%   \Alt{s+|\group|}{\xv}{\yv}.
%   $$
%   Hence, the invariant code is contained in a smaller alternant
%   code, corresponding to the eva\-luation of polynomials of lower degree.
% \end{remark}
%%% End cut ACr

\begin{proof}
  From (\ref{eq:description_alternant}), we have
  $$
  \Alt{s}{\xv}{\yv} =
  \cset{\frac{1}{y_i \loc{\xv}'(x_i)}f(x_i)}{i=1, \ldots, n}{f \in
  \Fqq [z]_{< n - s}} \cap \Fq^n.
  $$
  This code is obtained by evaluation of polynomials of degree up to
  $$n - s - 1 = (2^\gamma (n_0 - s_0) - 1).$$
  From
  Theorem~\ref{thm:invariant_poly}, the invariant codewords of
  $\Alt{s}{\xv}{\yv}$ come from evaluations of polynomials of the form
  $h \circ \psi_\group$.  Such polynomials have a degree that is a
  multiple of $\deg \psi_\group = 2^\gamma$ and hence their degree
  cannot exceed $2^\gamma (n_0 - s_0 - 1)$.  Thus, they
  should lie in
  $\Fqq [z]_{\leq n-s -|\group|} = \Fqq[z]_{<n-s-|\group|+1}$.  This
  leads to
  \begin{align*}
  \inv{\Alt{s}{\xv}{\yv}}{\group} &\subseteq
  \cset{\frac{1}{y_i \loc{\xv}'(x_i)} f(x_i)}{i = 1, \ldots, n}{
  f\in \Fqq[z]_{<n-s-|\group|+1}} \cap \Fq^n\\
                 & \subseteq \Alt{s + |\group| - 1}{\xv}{\yv}.
  \end{align*}
  This proves (\ref{item:a}).
  
  To prove (\ref{item:b}), note that the assumption on 
  $\Alt{s_0}{\overline{\psi_\group(\xv)}}{\overline {\yv}}$
  asserts the existence of
  $f \in \Fqq [z]_{< n_0 - s_0}$ such that $\deg f = n_0 - s_0 - 1$ and
  $f(\overline {\psi_\group (\xv)}) \in \Fq^{n_0}$.
  Thus, $f(\psi_\group (\xv)) \in \Fq^n$ and
  $\deg (f \circ \psi_\group) = n -s - |\group|$. Therefore
  $f (\psi (\xv)) \in \inv{\Alt{s}{\xv}{\yv}}{\group}$
  and $\inv{\Alt{s}{\xv}{\yv}}{\group}$ contains an
  element of
  $\Alt{s+ |\group| -1}{\xv}{\yv}$ that is not in
  $\Alt{s+|\group|}{\xv}{\yv}$.\qed
\end{proof}

% \begin{example}\label{exmp:1}
%   \begin{itemize}
%   \item   For \dags{1}, $\Cpub = \Alt{208}{\xv}{\yv}$ for some $\xv, \yv$ and
%     $|\group| = 16$. Next, $\inv{(\Cpub)}{\group} \subseteq \Alt{223}{\xv}{\yv}$.
%   \item For \dags{3}, $\Cpub = \Alt{352}{\xv}{\yv}$, $|\group| = 32$ and
%     $\inv{(\Cpub)}{\group} \subseteq \Alt{383}{\xv}{\yv}$.
%   \item For \dags{5}, $\Cpub = \Alt{704}{\xv}{\yv}$, $|\group| = 64$ and
%     $\inv{(\Cpub)}{\group} \subseteq \Alt{767}{\xv}{\yv}$.
%   \end{itemize}
% \end{example}

%%% Local Variables:
%%% mode: latex
%%% TeX-master: "Article"
%%% End:

\section{Presentation of the attack}\label{sec:attack}

\subsection{Context}
Recall that the extension degree is always $m = 2$.
The public code is the QD alternant code
$$
\Cpub \eqdef \Alt{r}{\xv}{\yv},
$$
with a permutation group $\group$ of cardinality
$|\group| = 2^\gamma$. As in
\S~\ref{ss:DAGS}, the code has a length $n = n_0 2^\gamma$,
dimension $k$ and is defined over a field $\Fq$ and
$q = 2^\ell$ for some positive integer $\ell$.
The degree $r$ of the alternant code is also a multiple of
$|\group| = 2^\gamma$ and hence is of the form $r = r_0 2^\gamma$.  We
suppose from now on that the classical lower bound on the dimension $k$
is reached, i.e. $k = n -2r$.
This always holds in the parameters proposed in
\cite{BBBCDGGHKONPR17}.  We finally set $k_0 = k/2^\gamma$.  In
summary, we have the following notation
\begin{equation}
  \label{eq:nota_0}
  n = n_0 2^\gamma, \quad k = k_0 2^\gamma, \quad r = r_0 2^\gamma.
\end{equation}
% We finish with some assumptions which turned out to be always true
% for the codes we considered.
% \begin{assumption}\label{ass:Cpub}
%   The code $\Cpub$ satisfies the following properties:
%   \begin{enumerate}[(i)]
%   \item\label{item:non_degenerate} the invariant code
%     $\inv{\Cpub}{\group}$ is non degenerate, i.e. any generator
%     matrix for this code has no zero column;
%   \item\label{item:distinct_alternant} the codes
%     $\Alt{r_0}{\overline{\psi_\group (\xv)}}{\overline{\yv}}$ and 
%     $\Alt{r_0+1}{\overline{\psi_\group (\xv)}}{\overline{\yv}}$
%     are distinct.
%     (See \S~\ref{ss:inv_alternant} for notation
%     $\psi_\group, \overline{\yv}$ and so on).
%   \end{enumerate}
% \end{assumption}

% \begin{remark}
%   Recall that, from Theorem~\ref{thm:invariant_alternant}, we have
%   $\inv{(\Cpub)}{\group} = \Alt{r_0}{\overline{\psi_\group
%       (\xv)}}{\overline{\yv}}$. Thus, the point
%   Assumption~\ref{ass:Cpub}(\ref{item:distinct_alternant}), is to assert
%   that $\inv{(\Cpub)}{\group}$ is not equal to a smaller alternant code.
% \end{remark}

\subsection{The subcode $\DC$}\label{ss:DC} 

We introduce a subcode $\DC$ of $\Cpub$
and prove that its knowledge permits to compute the norm
trace code. This code $\DC$ is unknown
by the attacker and we will see in \S~\ref{sec:complexity}
that the time consuming part of the attack
consists in guessing it.

\begin{definition}\label{def:code_D}
  Suppose that $|\group| \leq q$.
  We define the code $\DC$ as
  $$
  \DC \eqdef \inv{\Alt{r+q}{\xv}{\yv}}{\group}.
  $$
\end{definition}

\begin{remark}
  For parameters suggested in DAGS, we always have $|\group| \leq q$,
  with strict inequality for \dags{1} and \dags{3} and equality for
  \dags{5}. 
\end{remark}

\begin{remark}
The case $q < |\group|$ which never holds in DAGS suggested
parameters would be particularly easy to treat.
In such a situation, replacing possibly $\group$ by a subgroup,
one can suppose that $|\group| = 2q$.
Next, 
according to Theorem~\ref{thm:deg_inv},
and Heuristic~\ref{heur:strong}, we would have
$$
\cond{\inv{(\Cpub)}{\group}}{\Cpub} =
 \NTC (\xv),
$$
which would provide a very simple manner to compute $\NTC(\xv)$.  
\end{remark}

The following results are the key of the
attack. Theorem~\ref{thm:subcode_D} explains why this subcode $\DC$ is
of deep interest and how it can be used to recover the norm--trace
code, from which the secret key can be recovered (see
\S~\ref{ss:finishing}). Theorem~\ref{thm:codimD} explains why this
subcode $\DC$ can be computed in a reasonable time thanks to the QD
structure. Indeed, it shows that even if $\DC$ has a large codimension
as a subcode of $\Cpub$ its codimension in $\inv{(\Cpub)}{\group}$ is
much smaller. This is why the QD structure plays a crucial role in this
attack.

\begin{theorem}\label{thm:subcode_D}
%Recall that $q = 2^\ell$ and $|\group| = 2^\gamma$.
Under Heuristic~\ref{heur:strong} and assuming that
$\punctinv{\Alt{r+q}{\xv}{\yv}}{\group}$ is fully non degenerate
(see Definition~\ref{def:full_ndgn}), we have
$$
\cond{\DC}{\Cpub} = \NTC (\xv).
$$
\end{theorem}

\begin{proof}
  It is a direct consequence of Theorem~\ref{thm:deg_inv}
  and Heuristic~\ref{heur:strong}.\qed
\end{proof}

% That is, denoting $q$
% by $q = 2^\ell$, the codimension of $\DC$ in $\inv{(\Cpub)}{\group}$
% equals $2(\ell -  \gamma + 1)$ which equals $4$ for the parameters
% of \dags{1} and \dags{3} and equals to $2$ for the parameters of 
% \dags{5}.

% A consequence of these observations is that in practice \dags{5}
% is much easier to attack by this manner despite its security level
% is estimate to be above $256$ bits.

% The code $\DC$ of interest is the following one:
% $$
% \DC = \inv{\Alt{r+(\ell - \gamma + 1)|\group|}{\xv}{\yv}}{\group}.
% $$

\begin{theorem}\label{thm:codimD}
  The code $\DC$ has codimension
  $\leq \frac{2q}{|\group|} = 2^{\ell - \gamma +1}$ in
  $\inv{(\Cpub)}{\group}$.
\end{theorem}

\begin{proof}
  Using Theorem~\ref{thm:invariant_alternant}, we know that $\DC$ has
  the same dimension as
  $\Alt{r_0 + \frac{q}{|\group|}}{\overline{\psi_\group(\xv)}}{\overline{\yv}}$.
  This code has dimension
  $\geq n_0 - 2(r_0 + \frac{q}{|\group|})$.
  Since $\dim \inv{(\Cpub)}{\group} = k_0 = n_0 - 2r_0$, we get
  the result.\qed
\end{proof}

\begin{remark}
  Actually the codimension equals $2^{\ell - \gamma +1}$ almost all 
  the time.
\end{remark}

\renewcommand\arraystretch{1.3}
\begin{table}[h]
  \centering
  \begin{tabular}{|c|c|c|}
    \hline
    Proposal & $\DC$ & Codimension in $\inv{(\Cpub)}{\group}$ \\
    \hline \hline
    \dags{1} & $\inv{\Alt{240}{\xv}{\yv}}{\group}$ & 4 \\
    \hline
    \dags{3} & $\inv{\Alt{416}{\xv}{\yv}}{\group}$ & 4 \\
    \hline
    \dags{5} & $\inv{\Alt{768}{\xv}{\yv}}{\group}$ & 2 \\
    \hline
  \end{tabular}
  \vspace{.2cm}
  \caption{Numerical values for the code $\DC$}
  \label{tab:Num_values_for_DC}
\end{table}

% \begin{example}
%   \begin{itemize}
%   \item For \dags{1}, $\DC = \inv{\Alt{240}{\xv}{\yv}}{\group}$ and the code has
%     codimension $4$ in $\inv{(\Cpub)}{\group}$;
%   \item For \dags{3}, $\DC = \inv{\Alt{416}{\xv}{\yv}}{\group}$ and the code has
%     codimension $4$ in $\inv{(\Cpub)}{\group}$;
%   \item For \dags{5}, $\DC = \inv{\Alt{768}{\xv}{\yv}}{\group}$ and the code has
%     codimension $2$ in $\inv{(\Cpub)}{\group}$;
%   \end{itemize}
% \end{example}

\subsection{Description of the attack}
% In \S~\ref{ss:DC}, we introduce a subcode $\DC$ of codimension
% $\frac{2q}{|\group|}$ of $\inv{(\Cpub)}{\group}$ . This subcode
% $\DC$ is unknown, while its knowledge
% permits to recover $\NTC(\xv)$ as the conductor
% $\cond{\DC}{\Cpub}$. The difficult part of the attack
% consists in guessing this unknown code $\DC$.

The attack can be summarised as follows:
\begin{enumerate}[(1)]
\item Compute $\inv{(\Cpub)}{\group}$;
\item Guess the subcode $\DC$ of $\inv{(\Cpub)}{\group}$ of codimension
  $\frac{2q}{|\group|}$  such that
 $$\cond{\DC}{\Cpub} = \NTC (\xv);$$
\item Determine $\xv$ from $\NTC (\xv)$ and then $\yv$ from $\xv$.
\end{enumerate}

The difficult part is clearly the second one: how to guess $\DC$?
We present two manners to realise this guess.
\begin{itemize}
\item The first one consists in performing exhaustive
  search on subcodes of codimension $\frac{2q}{|\group|}$ of
  $\inv{(\Cpub)}{\group}$.
\item The second one consists in finding both $\DC$ and $\NTC (\xv)$
   by solving a system of equations of degree $2$ using Gröbner bases.
\end{itemize}
The first approach has a significant cost but which remains
far below the expected security level of DAGS proposed
parameters. For the second approach, we did not succeed to get a
relevant estimate of the work factor but its practical implementation
permits to break \dags{1} in about $20$ minutes and \dags{5} in less
than one minute (see \S~\ref{sec:implem} for further details on the
implementation).  We did not succeed to break \dags{3} parameters
using the second approach.  On the other hand the first approach
would have a work factor of $\approx 2^{80}$ for keys with an
expected security of $192$ bits.

The remainder of this section is devoted to detail the different
steps of the attack. %\ac{On peut mettre une outline ici... ou pas}

\subsection{First approach, brute force search of $\DC$}
\label{ss:brute_force}
A first way of getting $\DC$ and then of obtaining $\NTC (\xv)$
consists in enumerating all the subspaces
$\XC \subseteq \inv{(\Cpub)}{\group}$ of codimension
$\frac{2q}{|\group|}$ until we find one such that
$\cond{\XC}{\Cpub}$ has dimension $4$. Indeed, for an arbitrary $\XC$
the conductor will have dimension $1$ and be generated by $\onev$,
while for $\XC = \DC$ the conductor will be $\NTC(\xv)$ which has
dimension $4$.

The number of subspaces to enumerate is in
$O(q^{(2q/|\group|) (k_0 - 2q/|\group|)})$ which is in general
much too large to make the attack practical. It is however possible to reduce
the cost of brute force attack as follows.

\subsubsection{Using random subcodes of dimension
$2$}\label{sss:random_brute_force}
For any parameter set proposed in DAGS, the public
code has a rate $k/n$ less than $1/2$. Hence, its dual has rate larger than
$1/2$. Therefore, according to Heuristic~\ref{heur:strong},
given a random subcode $\DC_0$ of $\DC$ of dimension $2$, then 
$\cond{\DC_0}{\Cpub} = \NTC (\xv)$ with a high probability.

Thus, one can proceed as follows
\begin{itemize}
\item Pick two independent vectors $\cv , \cv' \in \inv{(\Cpub)}{\group}$
  at random
  and compute $\cond{\langle \cv, \cv' \rangle}{\Cpub}$;
\item If the conductor has dimension $4$, you probably found
  $\NTC (\xv)$, then pursue the attack as explained in
  \S~\ref{ss:finishing}.
\item Else, try again.
\end{itemize}
The probability that $\cv, \cv' \in \DC$ equals
$q^{-\frac{4q}{|\group|}}$. Therefore, one may have found $\NTC (\xv)$
after $O(q^{\frac{4q}{|\group|}})$ computations of conductors.

\begin{example}
  The average number of computations of conductors will be
  \begin{itemize}
  \item $O(q^8) = O(2^{40})$ for \dags{1};
  \item $O(q^8) = O(2^{48})$ for \dags{3};
  \item $O(q^4) = O(2^{24})$ for \dags{5}.
  \end{itemize}
\end{example}

\subsubsection{Using shortened codes}
Another manner consists in replacing the public code by one of its
shortenings. For that, we shorten $\Cpub = \Alt{r}{\xv}{\yv}$ at a set
of $a = a_0 2^\gamma$ positions which is a union of blocks, so that
the shortened code remains QD. We choose the integer $a$ such that the
invariant subcode
of the shortened code has dimension $2 + {\frac{2q}{|\group|}}$ and hence
the shortening of $\DC$ has dimension $2$.  Let $\Ind$ be
such a subset of positions. To determine $\sh{\DC}{\Ind}$, we
can enumerate any subspace $\XC$ of dimension $2$ of
$\sh{\Cpub}{\Ind}$ and compute $\cond{\XC}{\sh{\Cpub}{\Ind}}$. In
general, we get the trivial code spanned by the all--one codeword
$\onev$. If the conductor has dimension $4$ it is highly likely that
we found $\sh{\DC}{\Ind}$ and that the computed conductor equals
$\NTC (\xv_\Ind)$.

The number of such spaces we enumerate is in
$O (q^{\frac{4q}{|\group|}})$, which is very similar to
the cost of the previous method.

\subsection{Second approach, solving polynomial system of degree $2$}
\label{ss:Groebner}

An alternative approach to recover $\DC$ and $\NTC (\xv)$
consists in solving a polynomial system.
We proceed as follows.
Since $\tr (\xv) \in 
\cond{\DC}{\Cpub}$ and, from
Proposition~\ref{prop:formula_conductor},
$\cond{\DC}{\Cpub} = {(\spc{\DC}{\Cpub^\perp})}^\perp$, then
$$
\Gm_{\spc{\DC}{\Cpub^\perp}} \cdot \tr(\xv)^\top = 0,
$$
where $\Gm_{\spc{\DC}{\Cpub^\perp}}$ denotes a generator
matrix of $\spc{\DC}{\Cpub^\perp}$.
The above identity holds true when replacing $\tr(\xv)$ by
$\tr(\beta \xv)$ for any $\beta \in \Fqq$. Hence,
\begin{equation}\label{eq:identity_for_system}
\Gm_{\spc{\DC}{\Cpub^\perp}} \cdot \xv^\top = 0.
\end{equation}
The above identity provides the system we wish to solve.
We have two type of unknowns: the code $\DC$ and the vector $\xv$.
Set $c \eqdef \frac{2q}{|\group|}$ the codimension of $\DC$
in $\inv{(\Cpub)}{\group}$.
For $\DC$, let us introduce $(k_0-c)k_0$ formal variables
$U_{11}, \ldots, U_{1, c},$
$\ldots, U_{k_0-c, 1}, \ldots, U_{k_0-c, c}$ and set
$$
\mat U \eqdef 
\begin{pmatrix}
  U_{11} & \cdots & U_{1, c}  \\
  \vdots &        & \vdots    \\
  U_{k_0-c, 1} & \cdots & U_{k_0-c, c}
\end{pmatrix}
\qquad {\rm and}
\qquad
\Gm(U_{ij}) \eqdef
\begin{pmatrix}
  ~\mat I_{k_0 - c} ~|~ \mat U ~
\end{pmatrix}
\cdot \Gm^{\rm inv},
%\in \mathcal M_{k_0 - c, n}(\Fq [U_{ij}]).
$$
where $\mat I_{k_0 - c}$ denotes the $(k_0 - c) \times (k_0 - c)$
identity matrix and $\Gm^{\rm inv}$ denotes a $k_0 \times n_0$
generator matrix of $\inv{(\Cpub)}{\group}$.  It is probable that
$\DC$ has a generator matrix of the form $\Gm (u_{ij})$ for some
special values $u_{11}, \ldots, u_{k_0-c, c} \in \Fq$.  The case where
$\DC$ has no generator matrix of this form is rare and can be
addressed by choosing another generator matrix for
$\inv{(\Cpub)}{\group}$.

Now, let $\Hm$ be a parity--check matrix of $\Cpub$. A generator
matrix of $\spc{\DC}{\Cpub^\perp}$ can be obtained by constructing a
matrix whose rows list all the possible Schur products of one row of a
generator matrix of $\DC$ by one row of a parity--check matrix of
$\Cpub$. Therefore, let $\mat R(U_{ij})$ be a matrix with entries in
$\Fq [U_{1,1}, \ldots, U_{k_0-c, c}]$ whose rows list all the possible
Schur products of one row of $\Gm(U_{i,j})$ and one row of
$\Hm$. Hence, there is a specialisation
$u_{11}, \ldots, u_{k_0-c, c} \in \Fq$ of the variables $U_{ij}$
such that $\mat R (u_{ij})$ is a generator matrix of $\spc{\DC}{\Cpub^\perp}$.

The second set of variables $X_1, \ldots, X_n$
corresponds to the entries of
$\xv$. Using~(\ref{eq:identity_for_system}), the polynomial system we
have to solve is nothing but
\begin{equation}\label{system}
\mat R(U_{ij}) \cdot
\begin{pmatrix}
  X_1 \\ \vdots \\ X_n
\end{pmatrix}
= 0.
\end{equation}

\subsubsection{Reducing the number of variables}
Actually, it is possible to reduce the number of variables
using three different tricks.
\begin{enumerate}
\item Since the code is QD, the vector $\xv$ is a union of
  orbits under the action of the additive group $\group$.
  Therefore, one can introduce formal variables $A_1, \ldots, A_\gamma$
  corresponding to the generators of $\group$. Then, one can replace 
  $(X_1, \ldots, X_n)$ by
  \begin{equation}\label{eq:formal_support}
  (T_1,\ T_1 + A_1,\ \ldots\ ,\ T_1 + A_1 + \cdots + A_\gamma,\ 
  T_2, T_2+A_1,\ \ldots \ ).
  \end{equation}
  for some variables $T_1, \ldots, T_{n_0}$.
\item Without loss of generality and because of the $2$--transitive
  action of the affine group on $\Fqq$, one can suppose that
  the first entries of $\xv$ are $0$ and $1$ respectively
  (see for instance \cite[Appendix A]{COT17}).
  Therefore, in (\ref{eq:formal_support}), one can replace $T_1$ by $0$
  and $A_1$ by $1$.
\item Similarly to the approach of \S~\ref{ss:brute_force},
  one can shorten the codes so that
  $\DC$ has only dimension $2$, which reduces the number of variables $U_{ij}$
  to $2c$ and also reduces the length of the support we seek and
  hence reduces the number of the variables $T_i$.
\end{enumerate}

\subsubsection{On the structure of the polynomial system}
The polynomial equations have all the following features:
\begin{itemize}
\item Any equation is the sum of an affine and a bilinear form;
\item Any degree $2$ monomial is either
  of the form $U_{ij}A_k$ or of the form $U_{ij}T_k$. 
\end{itemize}

Table~\ref{tab:nr_of_variables} lists for the different proposals
the number of variables of type $U, A$ and $T$ of the system when
we use the previously described shortening trick.

\begin{table}[h]
  \centering
  \begin{tabular}{|c|c||c|c|}
    \hline
    Proposal & Number of $U_{ij}$'s & Number of $A_i$'s & Number of $T_i$'s \\
    \hline
    \dags{1} & 8 & 3 & 31 \\
    \hline
    \dags{3} & 8 & 4 & 27 \\
    \hline
    \dags{5} & 4 & 5 & 25 \\
    \hline
  \end{tabular}
  \vspace{.2cm}
  \caption{Number of variables of type $U, A$ and $T$ of the system}
  \label{tab:nr_of_variables}
\end{table}

\subsection{Finishing the attack}\label{ss:finishing}

When the previous step of the attack is over, then, if we used the
first approach based on a brute force search of $\DC$, we know at
least $\NTC (\xv)$ or $\NTC (\xv_\Ind)$ for some set $\Ind$ of
positions. If we used the second approach, then $\xv$ is already
computed, or at least $\xv_{\Ind}$ for some set of indexes $\Ind$.
Thus, there remains to be able to
\begin{enumerate}[(1)]
\item recover $\xv$ from $\NTC (\xv)$ or $\xv_\Ind$ from $\NTC (\xv_\Ind)$;
\item recover $\yv$ from $\xv$ or $\yv_\Ind$ from $\xv_\Ind$;
\item recover $\xv, \yv$ from $\xv_\Ind, \yv_\Ind$.
\end{enumerate}

%Let us treat these three questions.

\subsubsection{Recovering $\xv$ from $\NTC (\xv)$}
The code $\NTC (\xv)$ has dimension $4$ over $\Fq$
and is spanned by $\onev, \tr(\xv), \tr(\alpha \xv), \nr (\xv)$.
It is not difficult to prove that%  the same code after base field
% extension satisfies
$$
\NTC (\xv) \otimes \Fqq = \langle \onev, \xv, \xv^{\star q},
\xv^{\star(q+1)} \rangle,
$$
where $\NTC (\xv) \otimes \Fqq$ denotes the $\Fqq$--linear code
contained in $\Fqq^n$ and spanned over $\Fqq$ by the elements of
$\NTC (\xv)$.

% This code is peculiar in the sense that its square is
% smaller than the square of an arbitrary code of dimension $4$. Indeed,
% according to \cite{CCMZ15}, the square of a random code of dimension
% $4$ has dimension $10$ with a high probability, while:
% \begin{lemma}\label{lem:distinguish_NTC}
%   If $n > 2q+2$, then
%   $\dim \NTC (\xv)^{\star 2} = 9$.
% \end{lemma}

% \begin{proof}
%   Note first that
%   $$
%     \dim_{\Fq}\NTC (\xv)^{\star 2} = 
%   \dim_{\Fqq} (\NTC (\xv)^{\star 2})\otimes \Fqq =
%   \dim_{\Fqq} (\NTC (\xv) \otimes \Fqq)^{\star 2}.$$
%   Hence, let us study the square of $\NTC (\xv) \otimes \Fqq$,
%   \begin{align*}
%   \langle \onev, \xv, \xv^{\star q}, \xv^{\star(q+1)}\rangle^{\star 2}
%   & =\\
%     \langle \onev, \xv, \xv^{\star 2}, &\xv^{\star q}, \xv^{\star (q+1)},
%   \xv^{\star (q+2)}, \xv^{\star (2q)},  \xv^{\star (2q+1)},
%    \xv^{\star (2q+2)}\rangle.
%   \end{align*}
%   One can check that these vectors are independent when
%   $n > 2q+2$.\qed
% \end{proof}

% \begin{remark}
%   The previous lemma provides an additional test to check whether the
%   code we computed was actually $\NTC (\xv)$ by simply computing its
%   square.% It has been used in our implementations.
% \end{remark}

Because of the $2$--transitivity
of the affine group on $\Fqq$, without loss of generality,
one can suppose that the first entry
of $\xv$ is $0$ and the second one is $1$
(see for instance \cite[Appendix A]{COT17}).  Therefore, after
shortening $\NTC (\xv) \otimes \Fqq$ we get a code that we
call $\SC$, which is of the form
$$
\SC \eqdef \sh{\NTC (\xv)\otimes \Fqq}{\{1\}} = 
\langle \xv, \xv^{\star q}, \xv^{\star (q+1)} \rangle_{\Fqq}.
$$
Next, a simple calculation shows that
$$
\SC \cap \SC^{\star 2} = \langle \xv^{\star (q+1)}\rangle.
$$
Since, the second entry of $\xv$ has been set to $1$, we can 
deduce the value of $\xv^{\star (q+1)}$.

\begin{remark}\label{rem:field_of_def_for_computations}
  Actually, both $\SC$ and $\NTC (\xv)$ have a basis defined over
  $\Fq$, therefore, to get $\langle \xv^{\star (q+1)}\rangle_{\Fq}$
  it is sufficient to perform any computation on codes
  defined over $\Fq$.
\end{remark}

Now, finding $\xv$ is easy: enumerate the affine subspace of
$\NTC (\xv) \otimes \Fqq$ of vectors
whose first entry is $0$ and second entry is $1$
(or equivalently, the affine subspace of vectors of $\SC$
whose first entry equals $1$). For any such vector $\cv$, compute
$\cv^{\star (q+1)}$. If $\cv^{\star (q+1)} = \xv^{\star (q+1)}$, then $\cv$
equals either $\xv$ or $\xv^{\star q}$.
Since $\Alt{r}{\xv}{\yv} = \Alt{r}{\xv^{\star q}}{\yv^{\star q}}$
(see for instance \cite[Lemma 39]{COT17}), taking
$\xv$ or $\xv^{\star q}$ has no importance. Thus, without loss of generality,
one can suppose $\xv$ has been found.

\subsubsection{Recovering $\yv$ from $\xv$}\label{ss:yv_from_xv}
This is very classical calculation. The public code
$\Cpub$ is alternant, and hence is well--known to have a 
parity--check matrix defined over $\Fqq$ of the form
\begin{equation}\label{eq:Parity_check_alternant}
\Hm_{\rm pub} = \begin{pmatrix}
  y_1 & \cdots & y_n \\
  x_1 y_1 & \cdots & x_n y_n \\
  \vdots & & \vdots \\
  x_1^{r-1}y_1 & \cdots & x_n^{r-1}y_n
\end{pmatrix}.
\end{equation}
Denote by $\Gm_{\rm pub}$ a generator matrix of $\Cpub$.
Then, since the $x_i$'s are known, then the $y_i's$
can be computed by solving the linear system 
$$
\Hm_{\rm pub} \cdot \Gm_{\rm pub}^{\top} = 0.
$$

\subsubsection{Recovering $\xv, \yv$ from $\xv_\Ind, \yv_{\Ind}$}
After a suitable reordering of the indexes, one can suppose that
$\Ind = \{s, s+1, \ldots, n\}$. Hence, the entries
$x_1, \ldots, x_{s-1}$ of $\xv$ and $y_1, \ldots, y_{s-1}$ are
known. %Let us explain how to compute $x_s, y_s$.
Set $\Ind' \eqdef \Ind \setminus \{s\}$.
Thus, let $\Gm(\Ind')$ be a generator matrix of $\Alt{r}{\xv_{\Ind'}}{
\yv_{\Ind'}}$, which is nothing by $\sh{\Cpub}{\Ind'}$.
Using (\ref{eq:Parity_check_alternant}), we have
$$
\begin{pmatrix}
  y_1 & \cdots & y_s \\
  x_1 y_1 & \cdots & x_s y_s \\
  \vdots & & \vdots \\
  x_1^{r-1}y_1 & \cdots & x_s^{r-1}y_s  
\end{pmatrix}
\cdot \Gm(\Ind')  = 0.
$$
In the above identity, all the $x_i's$ and $y_i's$ are known
but $x_s, y_s$. The entry $y_s$ can be found by solving the linear system
$$
\begin{pmatrix}
  y_1 & \cdots & y_s
\end{pmatrix}
\cdot
\Gm (\Ind') = 0.
$$
Then, $x_s$ can be deduced by solving the linear system
$$
\begin{pmatrix}
  x_1y_1 & \cdots & x_sy_s
\end{pmatrix}
\cdot
\Gm (\Ind') = 0.
$$
By this manner, we can iteratively recover the entries
$x_{s+1}, \ldots, x_n$ and $y_{s+1}, \ldots, y_n$.
The only constraint is that $\Ind$ should be
small enough so that $\sh{\Cpub}{\Ind}$ is nonzero. But this always
holds true for the choices of $\Ind$ we made in the previous sections.

% \subsection{Summary of the attack}

% \ac{To do, avec un pseudo code d'algorithme}

\subsection{Comparison with a previous attack}
First, let us recall the attack on Wild Goppa codes over quadratic
extensions \cite{COT17}. This attack concerns some subclass of
alternant codes called {\em wild Goppa codes}. For such codes a
distinguisher exists which permits to compute a filtration of the
public code. Hence, after some computations, we obtain the subcode
$\Alt{r+q+1}{\xv}{\yv}$ of the public code $\Alt{r}{\xv}{\yv}$.  Then,
according to Heuristic~\ref{heur:strong}, the computation of a
conductor permits to get the code $\NTC(\xv)$. As soon as $\NTC (\xv)$
is known, the recovery of the secret is easy. Note that, the use of
the techniques of \S~\ref{ss:finishing} can significantly simplify
the end of the attack of \cite{COT17} which was rather technical.

We emphasise that, out of the calculation of $\NTC (\xv)$ by computing a
conductor which appears in our attack so that in \cite{COT17}, the two
attacks remain very different. Indeed, the way one gets a subcode whose
conductor into the public code provides $\NTC (\xv)$ is based in
\cite{COT17} on a distinguisher which does not work for general
alternant codes which are not Goppa codes. In addition, in the present
attack, the use of the permutation group is crucial, while it was useless
in \cite{COT17}.

%%% Local Variables:
%%% mode: latex
%%% TeX-master: "Article"
%%% End:

%\input{Extensions}
\section{Complexity of the first version of the
  attack}\label{sec:complexity}
As explained earlier, we have not been able to provide a complexity
analysis of the approach based on polynomial system solving. In
particular because the Macaulay matrix in degree $2$ of the system
turned out to have a surprisingly low rank, showing that this polynomial
system was far from being generic.  Consequently,
we limit our analysis to the first approach based on performing a
brute force search on the subcode $\DC$.

Since we look for approximate work factors, we will discuss an upper
bound on the complexity and not only a big $O$.

\subsection{Complexity of calculation of Schur products}
\label{ss:complexity_Schur}
A Schur product $\spc{\AC}{\BC}$ of two codes $\AC, \BC$
of length $n$ and respective dimensions $k_a, k_b$
is computed as follows.
\begin{enumerate}
\item Take bases $\av_1, \ldots, \av_{k_a}$ and 
  $\bv_1, \ldots, \bv_{k_b}$ of $\AC$ and $\BC$ respectively
  and construct a matrix $\mat M$ whose 
  rows are all the possible products $\spc{\av_i}{\bv_j}$,
  for $1 \leq i \leq k_a$ and $1 \leq j \leq k_b$.
  This matrix has $k_a k_b$ rows and $n$ columns.
\item Perform Gaussian elimination to get a reduced echelon form
  of $\mat M$.
\end{enumerate}
The cost of the computation of a reduced echelon form of a
$s \times n$ matrix is $ns\min(n,s)$ operations in the base field.
The cost of the computation of the matrix $\mat M$ is the cost of
$k_a k_b$ Schur products of vectors, i.e.
$n k_a k_b$ operations in the base field. This
leads to an overall calculation of the Schur product equal to
$$
n k_a k_b + nk_a k_b \min (n, k_a k_b)
$$
operations in the base field.  When $k_a k_b \geq n$, the cost of the
Schur product can be reduced using a probabilistic shortcut described
in \cite{CMP17}.
%%% Begin cut ACr
It consists in computing an $n \times n$ submatrix of
$\mat M$ by choosing some random subset of products
$\spc{\av_i }{\bv_j}$.
%%% End cut ACr
This permits to reduce the cost of computing
a generator matrix in row echelon form of $\spc{\AC}{\BC}$ to
$2n^3$ operations in the base field.

\subsection{Cost of a single iteration of the brute force search}
Computing the conductor $\cond{\XC}{\Cpub}$ consists in computing the
code ${(\spc{\XC}{\Cpub^\perp})}^\perp$. Since our attack consists in
computing such conductors for various $\XC$'s, one can compute a
generator matrix of $\Cpub^\perp$ once for good.  Hence, one can
suppose a generator matrix for $\Cpub^\perp$ is known. Then, according
to \S~\ref{ss:complexity_Schur}, the calculation of a generator matrix
of $\spc{\XC}{\Cpub^\perp}$ costs at most $2n^3$ operations in $\Fq$.
%%% Begin cut ACr
% Next,
% note that for most of the iterations, there is no need to deduce a
% generator matrix in reduced echelon form of
% ${(\spc{\XC}{\Cpub^\perp})}^\perp$, since it suffices to evaluate the
% dimension of $\spc{\XC}{\Cpub^\perp}$, which is immediate from the
% generator matrix in reduced echelon form.  If the dimension of the
% code is not the expected one, namely $n-\textrm{dim} \DC = n - 4$,
% then we skip to the next iteration.

% Hence, the overall cost of a single iteration of the brute force search is
% bounded above by $2n^3$ operations in $\Fq$.
%%% End cut ACr

\subsection{Complexity of finding $\DC$ and 
  $\NTC (\xv)$}
According to \S~\ref{ss:brute_force}, the average number of iterations
of the brute force search is $q^{2{\rm Codim} \DC}$, that is
$q^{\frac{4q}{|\group|}}$.  Thus, we get an overall cost of the first
step bounded above by
$$
2n^3 q^{\frac{4q}{|\group|}}\ {\rm operations\ in\ }\Fq.
$$
Since, $n = \Theta (q^2)$, we get a complexity in
$O(n^{3+\frac{2q}{|\group|}})$ operations in $\Fq$ for the computation
of $\NTC (\xv)$.

\subsection{Complexity of deducing $\xv, \yv$
from $\NTC (\xv)$}
A simple analysis shows that
the final part of the attack is negligible compared to the 
previous step.
%%% Begin cut ACr
Indeed,
\begin{itemize}
\item the computation of $\NTC (\xv)^{\star 2}$ costs $O(n^2)$ operations
  in $\Fq$ (because of Remark~\ref{rem:field_of_def_for_computations},
  one can perform these computations over $\Fq$)
  since the code has dimension $4$;
\item the computation of $\NTC (\xv)^{\star 2} \cap \NTC (\xv)$
  boils down to linear algebra and costs
  $O(n^3)$ operations in $\Fq$;
\item The enumeration of the subset of $\NTC (\xv) \otimes \Fqq$ of
  elements whose first entry is $0$ an second one is $1$ and
  computation of their norm costs $O(q^4 n) = O(n^3)$ operations in
  $\Fqq$.  Indeed the affine subspace of $\NTC (\xv) \otimes \Fqq$
  which is enumerated has dimension $2$ over $\Fqq$ and hence has
  $q^4$ elements, while the computation of the component wise norm of
  a vector costs $O(n)$ operations assuming that the Frobenius
  $z \mapsto z^q$ can be computed in constant time in $\Fqq$.
\item The recovery of $\yv$ from $\xv$ boils down to linear algebra
  and hence can also be done in $O(n^3)$ operations in $\Fqq$. If we
  have to recover $\xv, \yv$ from $\xv_\Ind, \yv_\Ind$, it can be done
  iteratively by solving a system of a constant number of equations,
  hence the cost of one iteration is in $O(n^2)$ operations
  in $\Fqq$.
\end{itemize}
Thus, the overall cost remains in $O(n^3)$
  operations in $\Fqq$.
%%% End cut ACr

\subsection{Overall complexity}
As a conclusion, % the second part of the attack is negligible compared
% to the first one. Hence,
the attack has an approximate work factor of
\begin{equation}\label{eq:work_factor}
2n^3q^{\frac{4q}{|\group|}} {\rm operations\ in\ }\Fq.
\end{equation}

\subsection{Approximate work factors of the first variant of
the attack on DAGS parameters}

We assume that operations in $\Fq$ can be done in
constant time. Indeed, the base fields of the public keys of DAGS
proposal are $\F_{32}$ and $\F_{64}$. For such a field, it is reasonable
to store a multiplication and inversion table.

Therefore, we list in Table~\ref{tab:work_factors} some approximate
work factors for DAGS according to~(\ref{eq:work_factor}). The second
column recalls the security levels claimed
in~\cite{BBBCDGGHKONPR17} for the best possible attack.  The last
column gives the approximate work factors for the first variant of our
attack.  
\begin{table}[!h]
  \centering
  \begin{tabular}{|c||c|c|}
    \hline
    {\bf Name} & {\bf Claimed security level} & {\bf Work factor of our
                                                attack}\\
    \hline \hline
    \dags{1} & 128 bits & $\approx 2^{70}$ \\
    \hline
    \dags{3} & 192 bits & $\approx 2^{80}$ \\
    \hline
    \dags{5} & 256 bits & $\approx 2^{58}$\\
    \hline
  \end{tabular}
  \vspace{.2cm}
  \caption{Work factors of the first variant of the attack}
  \label{tab:work_factors}
\end{table}
\renewcommand\arraystretch{1}

%%% Local Variables:
%%% mode: latex
%%% TeX-master: "Article"
%%% End:

\section{Implementation}\label{sec:implem}
Tests have been done using Magma~\cite{BCP97} on an
Intel\textsuperscript{\textregistered} Xeon 2.27 GHz.

Since the first variant of the attack had too significant costs to be
tested on our machines, we tested it on
the toy parameters \dags{0}. We performed 20 tests, which succeeded in
an average time of $2$ hours.

On the other hand, we tested the second variant based on solving a
polynomial system on \dags{1}, {\tt \_3} and {\tt \_5}. We have not
been able to break \dags{3} keys using this variant of the attack, on
the other hand about $100$ tests have been performed for \dags{1} and
\dags{5}. The average running times are listed in
Table~\ref{tab:timings}.
  % of the attack
  % for \dags{1} keys is about $19$ minutes and
  % for \dags{5} keys is about 35 seconds.  

% \begin{table}[!h]
%   \centering
%   \begin{tabular}{|c||c|}
%     \hline
%     {\bf Name} &  {\bf Average time} \\
%     \hline \hline
%     \dags{0} & 2 hours \\
%     \hline
%   \end{tabular}
%   \vspace{.2cm}
%   \caption{Average times for the first variant of the attack, on
%     $\approx XXXX$ tests.}
%   \label{tab:timings}
% \end{table}

\begin{table}[!h]
  \centering
  \begin{tabular}{|c||c|c|}
    \hline
    {\bf Name} & {\bf Claimed security level} & {\bf Average time} \\
    \hline \hline
    \dags{1} & 128 bits & 19 mn \\
    \hline
    \dags{5} & 256 bits & < 1 mn \\
    \hline
  \end{tabular}
  \vspace{.2cm}
  \caption{Average times for the second variant of the attack.}
  \label{tab:timings}
\end{table}

%%% Local Variables:
%%% mode: latex
%%% TeX-master: "Article"
%%% End:

%\input{Conclusion}

\section*{Acknowledgements}
The authors are supported by French {\em Agence nationale de la
  recherche} grants ANR-15-CE39-0013-01 {\em Manta} and
ANR-17-CE39-0007 {\em CBCrypt}.  Computer aided calculations have been
performed using software {\sc Magma} \cite{BCP97}.  The authors
express their deep gratitude to Jean-Pierre Tillich and Julien
Lavauzelle for very helpful comments.

\bibliographystyle{splncs04}
\bibliography{codecrypto}
\end{document}